\documentclass[letterpaper, 10 pt, conference]{ieeeconf}

\usepackage[usenames, dvipsnames]{color}
\usepackage{xcolor}
\usepackage{amssymb}
\usepackage{mathtools}
\usepackage{enumerate}
\usepackage{amsmath}
\usepackage{array}
\usepackage{algorithm}
\usepackage{subcaption}
\usepackage{balance}
\usepackage[noend]{algpseudocode}
\usepackage{cite}
\usepackage[normalem]{ulem}
\useunder{\uline}{\ul}{}
\usepackage{pdfpages}
\makeatletter
\newcommand*{\rom}[1]{\expandafter\@slowromancap\romannumeral #1@}
\newcommand*\rfrac[2]{{}^{#1}\!/_{#2}}
\newcommand{\He}{H\!e}
\makeatother
\newtheorem{theorem}{Theorem}
\newtheorem{example}{Example}
\newtheorem{lemma}{Lemma}
\newtheorem{remark}{Remark}
\newtheorem{definition}{Definition}

\newif\ifComments
\Commentstrue

\usepackage{graphicx}
\graphicspath{{Images/}}
\DeclareGraphicsExtensions{.pdf,.emf,.png}
\IEEEoverridecommandlockouts
\begin{document}

\title{On the Simulation of Polynomial NARMAX Models}
\author{Dhruv~Khandelwal,~Maarten Schoukens,~and~Roland~T\'oth \thanks{This research is supported by the Dutch Organization for Scientific Research (NWO, domain TTW, grant: 13852) which is partly funded by the Ministry of Economic Affairs.}
\thanks{All authors are with the Control Systems group, Department of Electrical Engineering, Eindhoven University of Technology, 5600 MB Eindhoven, The Netherlands
        {\tt\small D.Khandelwal@tue.nl}}}

\maketitle

\begin{abstract}

In this paper, we show that the common approach for simulation non-linear stochastic models, commonly used in system identification, via setting the noise contributions to zero results in a biased response. We also demonstrate that to achieve unbiased simulation of finite order NARMAX models, in general, we require infinite order simulation models. The main contributions of the paper are two-fold. Firstly, an alternate representation of polynomial NARMAX models, based on Hermite polynomials, is proposed. The proposed representation provides a convenient way to translate a polynomial NARMAX model to a corresponding simulation model by simply setting certain terms to zero. This translation is exact when the simulation model can be written as an NFIR model. Secondly, a parameterized approximation method is proposed to curtail infinite order simulation models to a finite order. The proposed approximation can be viewed as a trade-off between the conventional approach of setting noise contributions to zero and the approach of incorporating the bias introduced by higher-order moments of the noise distribution. Simulation studies are provided to illustrate the utility of the proposed representation and approximation method.

\end{abstract}


\section{Introduction}
\label{sec:intro}

In the field of data-driven modelling, prediction error has been the primary choice of identification criterion for many decades. This choice is well-justified due to its ease of use and the maximum-likelihood interpretation of the one-step ahead predictor with respect to the sum-of-squared error identification criterion \cite{astrom1979maximum}. This interpretation holds true under the assumption that the underlying noise process is zero-mean and that the true data-generating system can indeed be represented by the chosen model class. It has been observed (for example, in \cite{piroddi2003identification,ljung2001estimating}) that the presence of an auto-regressive component in the prediction model can allow a relatively simpler model to achieve optimistic prediction results for an otherwise complex dynamical system. This observation leads to the following question - given a set of data measured from the unknown dynamical system and an identified model, how can one ensure that the model not only predicts the measured data, but also captures the dynamical structure of the true system. This question has been long recognized in literature (c.f. \cite{ljung1995system}). A key challenge that is inherent to this question is that one should be able to assess the quality of the identified model across multiple model classes, for example, to compare the quality of a linear model with that of a non-linear model.

Model validation is an important ingredient of the system identification framework. A common approach in model validation is to order models based on a pre-defined measure of complexity, and then add levels of complexity in a model until the data can no longer \textit{invalidate} the model \cite{ljung1999model}. This approach is in line with the Occam's Razor principle that suggests that one should choose the simplest model that explains the data. However, the approach becomes quite cumbersome and computationally infeasible when identification is done in a black-box setting that spans across multiple model classes. Another common approach to model invalidation is to estimate the \textit{generalization error} of an identified model (i.e. the performance of an identified model on an independent data-set). When a sufficiently large amount of data is available, the user can set aside a part of the data that is not used during model estimation. The independent data-set can then be used to estimate the generalization error of all models identified based on the data used for estimation. This method is called cross-validation \cite{hastie2001model}. When the amount of data is insufficient, information criterion such as AIC (Akaike Information Criterion) or BIC (Bayesian Information Criterion) can be used to estimate the generalization error.

An alternative approach that deals with this challenge is the control of complexity by using regularization. In these methods, the cost function is modified to include a penalty term that controls the complexity of the identified models as part of the optimization problem. However, in a black box setting that spans multiple model classes, finding a suitable regularization penalty is challenging. 

A third alternative, and one that is of interest in this contribution, is to make use of simulation error to judge the quality of an identified model. Simulation models do not make use of the past output measurements, as opposed to prediction models.  As a result, simulation error is typically more sensitive to mismatch between the model and the system structure than prediction error (see \cite{aguirre2010prediction, piroddi2003identification}). Moreover, simulation models can be defined for models belonging to different model classes. Simulation error has been used in the past for non-linear structure selection, for example in \cite{piroddi2003identification}.

The notion of a simulation model can be understood in multiple ways. A simulation model can be viewed as an infinite horizon prediction model. This effectively negates the auto-regressive components of a prediction model. A simulation model can also be interpreted as the ``deterministic response'' of the system. As a result of this interpretation, simulation models are often reported in the literature as prediction models with noise terms set to 0 \cite[ch. 5]{Ljung1999}. In the case of linear systems, it is always possible to lump the noise sources as an additive term on the output of the deterministic response of the system. Hence, the approach of setting noise contributions to 0 can be intuitively linked to the idea of obtaining the ``deterministic response'' of the system. However, in the non-linear setting, neglecting the process noise terms leads to a biased response (see \cite{hagenblad2008maximum,giordano2016consistency}).

In this paper, two contributions are made with respect to the computation of simulation models for polynomial NARMAX models. In Section \ref{sec:NARMAXsimple}, a Hermite polynomial-based representation is proposed. The proposed representation provides a convenient way to link prediction models and simulation models, under the assumption that the noise is distributed normally. In Section \ref{sec:NARMAXfull}, it is shown that computing a simulation model for a finite order NARMAX model may lead to an infinite order NIIR model. For such cases, an approximation method is proposed. Numerical illustrations are provided in \ref{sec:NumericalIllustration}.

\section{Problem description}

Consider the discrete-time single-input single-output (SISO) polynomial NARMAX model form:
\begin{multline}
	y_k = f\big( y_{k-1}, \dots, y_{k-n_y}, u_{k-1}, \dots, u_{k-n_u}, \\
	\quad \xi_{k-1}, \dots, \xi_{k-n_\xi} \big) + \xi_k, \label{eq:NARMAX}
\end{multline}
where, $f(.)$ is a polynomial non-linear function, $u_k, y_k \in \mathbb{R}$ are the input and the output of the system at time instant $k$, $\xi$ is a zero-mean white noise process assumed to be independent of $u$, and $n_u,n_y$ and $n_\xi$ are the maximum time lags for the respective signals. In the sequel it is assumed that $f$ in \eqref{eq:NARMAX} is Bounded Input Bounded Output (BIBO) stable with respect to the deterministic input $u$ and that the initial conditions $\{u_k, y_k, \xi_k\}$ for $k\le0$ are 0. Note that \eqref{eq:NARMAX} is not the most general representation of non-linear systems, nonetheless, a large family of systems can be modelled within this class. 

Given a model of form \eqref{eq:NARMAX}, one possible approach to compute the corresponding simulation model is to set the noise contributions to 0. This approach works well when the function $f$ is linear because it is always possible to separate the noise contributions in the function $f$ to form an additive noise model of the form
\begin{multline}
	y_k = g(y_{k-1}, \dots, y_{k-{n_y}}, u_{k}, \dots, u_{k-{n_u}}) + \\
	\quad h(\xi_{k-1}, \dots, \xi_{k-{n_\xi}}) + \xi_k,
\end{multline}
where $g$ and $h$ are linear functions. If a simulation model is described as the deterministic response of the system, then it can be expressed as the conditional expectation of the output $E[y_k]$ taken with respect to the distribution of the noise. Since it is assumed that the noise process is zero-mean, the simulation response $y_{\mathrm{s}}$ can be written as
\begin{equation}
	y_{\mathrm{s},k} \coloneqq E[y_k] = g(y_{\mathrm{s},k-1}, \dots, y_{\mathrm{s},k-{n_y}}, u_{k}, \dots, u_{k-{n_u}}).
\end{equation}
Notice that this is the same as setting all noise contributions in \eqref{eq:NARMAX} to 0. However, when $f$ is a non-linear function, it is not trivial to separate the noise contributions into a separate noise process that appears additively in the output of the system. Furthermore, in order to compute the conditional expectation of \eqref{eq:NARMAX}, one must take the expectation over terms that include polynomial factors of the past outputs $y$ and noise $\xi$. In order to compute such a conditional expectation, one must take in to account the higher order moments of the noise process and the non-linear dependencies of the output terms $y$ on the noise terms $\xi$. As a result, the simple approach of setting the noise terms to 0 results in a biased simulation response.
%

Continuing the line of reasoning from the linear case, one can represent the simulation response  $y_{\mathrm{s}}$ of \eqref{eq:NARMAX} as:
%
\begin{multline}
	y_{\mathrm{s},k} \coloneqq E[y_k] = E [ f( y_{k-1}, \dots, y_{k-n_y}, u_{k-1}, \dots u_{k-n_u}, \\
	 	\quad \xi_{k-1}, \dots, \xi_{k-n_\xi} ) ] \label{eq:NARMAXsim}
\end{multline}
where the conditional expectation is taken with respect to the probability density $p_\xi$ of $\xi$ and the input sequence $\{u_k;k\in [1,N]\}$. The problem can be stated as follows. Given an input-output sequence $\{u_k;k\in [1,N]\}$, and a model of form \eqref{eq:NARMAX}, compute the simulation model, as given in \eqref{eq:NARMAXsim}.

It is important to note that in order to compute simulation models as in \eqref{eq:NARMAXsim}, higher order moments of the noise distribution must be known. Hence, in the sequel, we assume that the noise process $\xi_k$ is i.i.d and is distributed as $\mathcal{N}(0,1)$. Note that a noise process $v_k \sim \mathcal{N}(\mu, \sigma)$, can be equivalently expressed as
	\begin{equation}
		v_k = \sigma \xi_k + \mu, \label{eq:nonStandardGaussian}
	\end{equation}
and hence, can be included in the model formulation in \eqref{eq:NARMAX}.


\section{Mathematical background}
\label{sec:math}

Under the given assumptions that the noise process is distributed normally, the noise process terms $\xi_k$ that appear in model \eqref{eq:NARMAX} may be transformed to Hermite polynomials \cite[Ch. 18]{olver2010nist} in $\xi_k$. For completeness, a brief overview of the relevant properties of Hermite polynomials is presented.

Let $(\mathbb{R}, \mathcal{B}_\mathbb{R}, p_x)$ be a probability space on the real line $\mathbb{R}$, where $\mathcal{B}_\mathbb{R}$ is the Borel $\sigma$-algebra on $\mathbb{R}$, and $p_x$ is the probability measure of the Gaussian distribution. Consider the Hilbert space $L_2(\mathbb{R}, \mathcal{B}_\mathbb{R}, p_x)$ of random functions $f$ that satisfy,
\begin{equation}
	E\left[f\right] = 0 \ \text{and} \ E\left[f^2\right]< \infty.
\end{equation}
where $E$ is associated with the inner product
\begin{equation}
	\langle f,g \rangle = E\left[f g \right] = \frac{1}{\sqrt{2\pi}}\int_{-\infty}^{\infty} f(x) g(x) e^{-\frac{x^2}{2}} \mathrm{d}x, \label{eq:innerProduct}
\end{equation}
for any $f,g \in L_2(\mathbb{R}, \mathcal{B}_\mathbb{R}, p_x)$.

\begin{definition}
	Hermite polynomials $\He_n(x)$ of degree $n$ are defined as
	\begin{equation}
		\He_n(x) = (-1)^n e^{\frac{x^2}{2}} \frac{\mathrm{d}^n}{\mathrm{d}x^n}e^{-\frac{x^2}{2}}. \label{eq:HermiteDef}
	\end{equation}
\end{definition}
\vspace*{0.15cm}
Hermite polynomials form a closed and complete orthogonal basis in the $L_2$ space with the inner product defined in \eqref{eq:innerProduct} (see \cite{olver2010nist}), i.e., 
\begin{align}
	\langle \He_n,\He_m \rangle &=  \frac{1}{\sqrt{2\pi}} \int_{-\infty}^{\infty} \He_n(x) \He_m(x) e^{-\frac{x^2}{2}} \mathrm{d}x \nonumber \\
	&= n!\delta_{n,m}, \label{eq:innerProdHermite}
\end{align}
where $\delta_{n,m} = 1$ when $n=m$ and $0$ otherwise. Hermite polynomials possess a number of properties that will be used in order to derive simulation models automatically from a given stochastic model of form \eqref{eq:NARMAX}.

Since Hermite polynomials form a complete orthogonal basis in $L^2(\mathbb{R}, \mathcal{B}_\mathbb{R}, p_x)$ with inner product defined in \eqref{eq:innerProduct}, any function $f \in L^2(\mathbb{R}, \mathcal{B}_\mathbb{R}, p_x)$ can be expressed as a convergent linear combination of basis functions

\vspace{-0.2cm}
{\small
\begin{equation}
	f = \sum_{n=0}^\infty \frac{1}{n!} \langle f,\He_n \rangle \He_n.
\end{equation}	 
}
\vspace{-0.2cm}

Given a monomial in $x\in \mathbb{R}$, one can compute an equivalent representation in terms of Hermite polynomials $\He(x)$ using the following inverse relationship

\vspace{-0.2cm}
{\small
\begin{equation}
	x^n = n! \sum_{m=0}^{\lfloor \frac{n}{2} \rfloor} \frac{1}{2^m m! (n-2m)!} \He_{n-2m}(x), \label{eq:hermiteInverse}
\end{equation}
}
\vspace{-0.2cm}

where $\lfloor . \rfloor$ is defined as the \emph{floor} operator, i.e., $\lfloor n \rfloor = q$, where $q \in \mathbb{Z}$ is the largest integer such that $q \leq n$.
Hermite polynomials of normally distributed random variables have useful properties. Let $X$ be a random variable distributed as $\mathcal{N}(0,1)$. It can be shown that

\vspace{-0.2cm}
{\small
\begin{equation}
	E[\He_n(X)] = \begin{cases}
	1 \quad \text{for } n=0,\\
	0 \quad \text{for } n \neq 0.
	\end{cases}
	\label{eq:expectationHermite}
\end{equation}
}
\vspace{-0.3cm}

The following result can be obtained by using \eqref{eq:innerProdHermite}

\vspace{-0.3cm}
{\small
\begin{align}
E[\He_n(X),\He_m(X)] &= \int_{-\infty}^{\infty} \He_n(x) \He_m(x) \frac{1}{\sqrt{2\pi}} e^{-\frac{x^2}{2}} \mathrm{d}x, \nonumber \\
&= \langle \He_n,\He_m \rangle, \nonumber \\
&= n!\delta_{n,m}. \label{eq:HermiteAutoCorr}
\end{align}
}
\vspace{-0.3cm}

It has been shown in \cite{bershad2001analysis} that for $X_1,X_2 \sim \mathcal{N}(0,1)$,
\begin{equation}
	E[\He_n(X_1),\He_m(X_2)] = n!\delta_{n,m} (E[X_1X_2])^n. \label{eq:HermiteCrossCorr}
\end{equation}
In particular, if $X_1,X_2$ are independent, we get

\vspace{-0.3cm}
{\small
\begin{equation}
	E[\He_n(X_1),\He_m(X_2)] = 
	\begin{cases} 
		1 \quad \text{if } n =0, m=0,\\
		0 \quad \text{otherwise.}
	\end{cases}
	\label{eq:HermiteCrossCorrIndependent}
\end{equation}
}
\vspace{-0.5cm}

\section{Solution approach}

\subsection{Polynomial NARMAX - simplified case}
	\label{sec:NARMAXsimple}
	
	We first consider the simpler problem of computing a simulation model for a sub-class of NARMAX models. Consider models of the form
	\begin{multline}
	y_k = f\left(u_{k-1}, \dots, u_{k-n_u}, \xi_{k-1}, \dots, \xi_{k-n_\xi} \right) + \\
	\quad g(y_{k-1}, \dots, y_{k-n_y}) + \xi_k, \label{eq:NARMAXsimple}
	\end{multline}
	where the function $f$ is \emph{polynomial} function of the regressor variables and $g$ is a \emph{linear} function of the past outputs. This is a specific class of models that can be represented as \eqref{eq:NARMAX}. Recall that $\xi_k \sim \mathcal{N}(0,1)$. Using \eqref{eq:hermiteInverse}, \eqref{eq:NARMAXsimple} can be reformulated by expressing all noise polynomial terms in the model in terms of Hermite polynomials, i.e.,
	
	\vspace*{-0.4cm}
	{\small
	\begin{align}
		y_k &= f\left(u_{k-1}, \dots, u_{k-n_u}, \xi_{k-1}, \dots, \xi_{k-n_\xi} \right) + \nonumber\\
		& \quad \quad \quad  g(y_{k-1}, \dots, y_{k-n_y}) + \xi_k, \nonumber \\
		&= \tilde{f}\big(u_{k-1}, \dots, u_{k-n_u}, \He_0(\xi_{k-1}), \dots, \He_{\bar{d}_1}(\xi_{k-1}), \dots, \nonumber\\
		& \quad \He_0(\xi_{k - n_{\xi} }), \dots, \He_{\bar{d}_{n_\xi}}(\xi_{k-n_\xi}) \big) + g(y_{k-1}, \dots, y_{k-n_y}) \nonumber\\
		&\quad + \xi_k, \label{eq:NARMAXsimpleHermite}
	\end{align}
	}
	\vspace*{-0.4cm}
	
	\noindent where $\bar{d}_i$ is the maximum exponent of the $\xi_{k-i}$ term, and the function $\tilde{f}$ is multi-linear in terms of the Hermite polynomials $\He_n(\cdot)$. The function $\tilde{f}$ will be later derived in Lemma \ref{lem:1}. This results in an equivalent model that can be easily transformed to a simulation model by using \eqref{eq:expectationHermite}, \eqref{eq:HermiteAutoCorr} and \eqref{eq:HermiteCrossCorrIndependent}. This yields a simulation model represented as

	\vspace{-0.4cm}
	{\small
	\begin{align}
		y_{\mathrm{s},k} &= E\big[ f\left(u_{k-1}, \dots, u_{k-n_u}, \xi_{k-1}, \dots, \xi_{k-n_\xi} \right) + \nonumber \\
		& \quad \quad \quad g(y_{k-1}, \dots, y_{k-n_y}) + \xi_k \big], \nonumber\\
		&= E\Big[ \tilde{f}\big(u_{k-1}, \dots, u_{k-n_u}, \He_0(\xi_{k-1}), \dots, \He_{\bar{d}_1}(\xi_{k-1}), \nonumber\\
		& \quad \quad \quad \dots, \He_0(\xi_{k - n_{\xi} }), \dots, \He_{\bar{d}_{n_\xi}}(\xi_{k-n_\xi}) \big) + \nonumber \\
		& \quad \quad \quad g(y_{k-1}, \dots, y_{k-n_y}) \Big], \nonumber\\
		&= f_{\mathrm{s}}(u_{k-1}, \dots, u_{k-n_u}, y_{\mathrm{s},k-1}, \dots, y_{\mathrm{s},k-n_y}). \label{eq:NARMAXsimpleSim}
	\end{align}		
	}
	\vspace{-0.5cm}
	
	It should be noted that an equation of form \eqref{eq:NARMAXsimple} can be easily transformed to the equivalent simulation model as in \eqref{eq:NARMAXsimpleSim} without the use of Hermite polynomials, since the moments of a Gaussian distribution are known. However, the use of Hermite polynomials yields an equivalent prediction model, as in \eqref{eq:NARMAXsimpleHermite}, that offers a convenient representation. It turns out that, in order to obtain a simulation model like \eqref{eq:NARMAXsimpleSim}, one must only ``switch off'' a number of Hermite polynomial terms in \eqref{eq:NARMAXsimpleHermite}, as per Equations \eqref{eq:expectationHermite}, \eqref{eq:HermiteAutoCorr} and \eqref{eq:HermiteCrossCorrIndependent}). An example will be used to illustrate the method to compute the transformation, followed by a general result.
	
	\begin{example}
	\label{ex:NARMAXsim}
		Consider a system described by the following equation
		\begin{equation}
			y_k = f(u,\xi) = u_k + \xi_{k-1}^2 + \xi_k, \label{eq:ex2}
		\end{equation}
		where $\{u_k\}$ and $\{ \xi_k \}$ are i.i.d sequences distibuted as $\mathcal{N}(0,1)$. The given equation can be re-written as follows
		
		\vspace{-0.5cm}
		{\small
		\begin{align}
			y_k &= u_k + 2! \sum_{m=0}^1 \frac{1}{2^m m! (2-2m)!} \He_{2-2m}(\xi_{k-1}) + \xi_k,  \nonumber \tag{using \eqref{eq:hermiteInverse}}\\
			&= u_k + 2\left( \frac{\He_2(\xi_{k-1})}{2} + \frac{\He_0(\xi_{k-1})}{2} \right)  + \xi_k, \nonumber\\
			&= u_k + \He_0(\xi_{k-1}) + \He_2(\xi_{k-1})  + \xi_k, \nonumber \\ 
			&= \tilde{f}\left( u_k,\He_0(\xi_{k-1}),\He_2(\xi_{k-1}) \right) + \xi_k \label{eq:ex2HermiteVer}.
		\end{align}
		}		
		
		\vspace{-0.5cm}
		Now, the simulation model can be computed by taking the expectation of \eqref{eq:ex2HermiteVer} with respect to the noise process $\xi$,
		\vspace{-0.4cm}
		
		{\small
		\begin{align}
			y_{\mathrm{s},k} &= E[y_k] = u_k + E[\He_0(\xi_{k-1}) + \He_2(\xi_{k-1})] + E[\xi_k], \nonumber  \\
			&= u_k + 1. \tag{using \eqref{eq:expectationHermite}} \\	
			&= f_{\mathrm{s}}(u_k) \label{eq:ex2SimModel}
		\end{align}
		}
		Observe that if the simulation model was computed by setting the contributions of the noise terms to 0, one would obtain the following simulation model
		\begin{equation}
			y_{\mathrm{s},k}' = u_k. \label{eq:ex1Sim0}
		\end{equation}
		The two simulation models are offset by a scalar factor 1.		
	\end{example}
	 In representation \eqref{eq:NARMAXsimpleHermite} (and \eqref{eq:ex2HermiteVer}), one must only set all non-zero order Hermite polynomials of the noise term to 0 in order to obtain a simulation model as in \eqref{eq:NARMAXsimpleSim} and \eqref{eq:ex2SimModel}.
	
	For a general result, the following equivalent representation of \eqref{eq:NARMAXsimple} will be used
	
	\vspace{-0.25cm}
	{\small
	\begin{equation}
			y_k = \sum_{i=1}^p \left( c_i \prod_{j=1}^{n_u} u_{k-j}^{b_{i,j}} \prod_{q=1}^{n_\xi} \xi_{k-q}^{d_{i,q}} \right) + \sum_{r=1}^{n_y} g_r y_{k-r} + \xi_k, \label{eq:NARMAXsimpleAlt}
	\end{equation}
	}%
	where  $p \in \mathbb{Z}_+$ is the number of monomial terms of $u$ and $\xi$ in the model, $c_i \in \mathbb{R}$ are the co-efficients, $b_{i,j}$ and $d_{i,q}$ are the exponents of the $u_{k-j}$ and $\xi_{k-q}$ factors in the $i^{\mathrm{th}}$ monomial term respectively and $g_r$ is the linear co-efficient of the $r^{\mathrm{th}}$ output term $y_{k-r}$.
	\begin{lemma}
	\label{lem:1}
		For a model of form \eqref{eq:NARMAXsimpleAlt}, and under the assumption that $\xi_k \sim \mathcal{N}(0,1)$ is independent of the input $u_l$ for all $l \in [1,N]$, the simulation model, computed in the sense of \eqref{eq:NARMAXsim}, is given by
		
		\vspace{-0.25cm}
		{\small
		\begin{equation}
			y_{\mathrm{s},k} = \sum_{{i \in P}_e} c_i \prod_{j=1}^{n_u} u_{k-j}^{b_{i,j}}\prod_{q=1}^{n_\xi}(d_{i,q}-1)!!  + \sum_{r=1}^{n_y} g_r y_{s_{k-r}}, \label{eq:resNARMAXsimple}
		\end{equation}
		}
		\vspace{-0.15cm}

		\noindent where $P_e\coloneqq \left\{ i \in [1,p] \mid d_{i,q} \text{ is even }\forall q \in [1,n_\xi] \right\}$ and $(n-1)!! \coloneqq \frac{n!}{\frac{n}{2}!2^{\rfrac{n}{2}} }$.
	\end{lemma}
	\vspace{0.1cm}
	\begin{proof}
		The proof relies on the use of the proposed Hermite polynomial based representation. Let $\gamma_i \coloneqq  c_i \prod_{j=1}^{n_u} u_{k-j}^{b_{i,j}}$.  Using \eqref{eq:hermiteInverse}, Equation \eqref{eq:NARMAXsimpleAlt} can be written as

		\vspace{-0.4cm}
		{ \small
		\begin{multline}
			y_k = \sum_{i=1}^{p} \Bigg( \gamma_i \prod_{q=1}^{n_\xi} \Bigg( d_{i,q}! \sum_{m=0}^{\lfloor \frac{d_{i,q}}{2} \rfloor} \frac{\He_{d_{i,q}-2m} (\xi_{k-q})}{2^m m! (d_{i,q}-2m)!}  \Bigg) \Bigg) + \\
			\quad  \sum_{r=1}^{n_y} g_r y_{k-r}. \label{eq:proposedRepSim}
		\end{multline}
		}
		\vspace{-0.30cm}
		
		Note that \eqref{eq:proposedRepSim} is of the form \eqref{eq:NARMAXsimpleHermite}. Define partitions $P_e\coloneqq \left\{ i \in [1,p] \mid d_{i,q} \text{ is even } \forall q \in [1,n_\xi] \right\}$ and $P_\mathrm{o} \coloneqq \left\{ i \in [1,p] \mid i \notin P_\mathrm{e} \right\}$. This yields
		
		\vspace*{-0.5cm}
		{\small
		\begin{multline}
			y_k = \sum_{{i \in P_\mathrm{e}}} \gamma_i \prod_{q=1}^{n_\xi} \left( d_{i,q}! \sum_{m=0}^{\lfloor \frac{d_{i,q}}{2} \rfloor} \frac{\He_{d_{i,q}-2m} (\xi_{k-q})}{2^m m! (d_{i,q}-2m)!} \right) + \\
			 	\quad \sum_{{i\in P_\mathrm{o}}} \gamma_i \prod_{q=1}^{n_\xi} \left( d_{i,q}! \sum_{m=0}^{\lfloor \frac{d_{i,q}}{2} \rfloor} \frac{\He_{d_{i,q}-2m} (\xi_{k-q})}{2^m m! (d_{i,q}-2m)!} \right) + \\
			 	\quad \sum_{r=1}^{n_y} g_r y_{k-r}.
		\end{multline}
		} 
		Taking the expectation with respect to the noise process $\xi$, and recognizing that the second term on the right hand side drops out (due to \eqref{eq:expectationHermite}), we get the following
		
		{\small
		\begin{align}
			y_{\mathrm{s},k} &= \begin{multlined}[t]
				\sum_{{i \in P}_\mathrm{e}} \gamma_i \prod_{q=1}^{n_\xi} E\Bigg( d_{i,q}! \bigg( \frac{\He_0 (\xi_{k-q})}{2^{\frac{d_{i,q}}{2}} \frac{d_{i,q}}{2}! } + \\
				\sum_{m=0}^{\lfloor \frac{d_{i,q}}{2} -1 \rfloor} \frac{\He_{d_{i,q}-2m} (\xi_{k-q})}{2^m m! (d_{i,q}-2m)!)}  \bigg) \Bigg) +  E\left[\sum_{r=1}^{n_y} g_r y_{k-r}\right],
			\end{multlined} \nonumber\\
			&= \sum_{{i \in P}_\mathrm{e}} \gamma_i \prod_{q=1}^{n_\xi}(d_{i,q}-1)!! + \sum_{r=1}^{n_y} g_r y_{s_{k-r}},
		\end{align}
		}
		
		\noindent which is the desired result.
	\end{proof}
	Notice that in the derivation of \eqref{eq:resNARMAXsimple}, no approximations were made, and hence the result is exact, as per the definition of the simulation model in \eqref{eq:NARMAXsim}. However, this will not be the case for the general polynomial NARMAX model, as will be shown in the next Section.	
	
	
	\subsection{General polynomial NARMAX models}
	\label{sec:NARMAXfull}
	
	We now consider the full polynomial NARMAX model as shown in \eqref{eq:NARMAX}. The procedure to compute a simulation response of a stochastic model proposed in Section \ref{sec:NARMAXsimple} allows us to compute the expectation of the noise terms in \eqref{eq:NARMAX} systematically. However, the proposed method does not deal with the non-linear dependence of the output signal $y$ on the noise process $\xi$. Although $\{u_k\}$ is a known sequence and $\xi_k$ is distributed normally, the random variable $y_k$ is typically not distributed normally. Hence, model terms involving random variables $y_k$ cannot be equivalently represented in terms of Hermite polynomials of $y_k$ (the inner-product in \eqref{eq:innerProdHermite} no longer corresponds to the expectation with respect to the distribution of $y_k$). Consequently, in order to compute the conditional expectation of \eqref{eq:NARMAX}, one must recursively eliminate all $y_k$ terms, which yields a non-linear infinite impulse response (NIIR) in terms of the input. Computing an NIIR is intractable because next to infinite time lags, it also contains infinite order polynomial exponents. As a result, just like in the LTI case, an approximation must be made in order to keep the simulation model tractable. Several approximation concepts are proposed in the following example.
	
	\begin{example}
	\label{ex:fullNarmax}
	Consider a system governed by the following equation
		\begin{equation}
			y_k = u_k - 0.1 y_{k-1}^2 + \xi_k, \label{eq:ex3model}
		\end{equation}
		where, $u_k \sim \mathcal{N}(\mu_u,\sigma_u)$ and $\xi_k \sim \mathcal{N}(0,1)$. By taking the expectation and recursively replacing the $y_{k-1}$ term in the model, we get the following set of equations
		{\small
		\begin{align*}
			E[y_k] &= E\left[ u_k - 0.1y_{k-1}^2 + \xi_k \right] = u_k - 0.1E[y_{k-1}^2],\\
			&= u_k - 0.1 E\left[ (u_{k-1} - 0.1y_{k-2}^2 + \xi_{k-1})^2 \right],\\
			&= \begin{multlined}[t]
				u_k - 0.1u_{k-1}^2 - 0.1\sigma_\xi^2 - 0.001 E[y_{k-2}^4] + \\
				\quad 0.02u_{k-1}E[y_{k-2}^2],
				\end{multlined}\\
			&= \begin{multlined}[t]
				u_k - 0.1u_{k-1}^2 - 0.1 - 0.001 E[y_{k-2}^4] +\\
				\quad 0.02 u_{k-1} E[(u_{k-2} - 0.1y_{k-3}^2 + \xi_{k-2})^2],
				\end{multlined}\\
			&= \begin{multlined}[t]
			u_k - 0.1u_{k-1}^2 - 0.1 - 0.001 E[y_{k-2}^4]\\
			\quad  + 0.02 u_{k-1} u_{k-2}^2 + 0.0002 u_{k-1} E[y_{k-3}^4] + \\
			\quad 0.02 u_{k-1} \sigma_\xi^2 - 0.004 u_{k-2}E[y_{k-3}^2],
			\end{multlined}\\
			&= \dots
		\end{align*}
		}
		From these equations, it can be seen that realization of the simulation model results in an NIIR. 
		We explore three approximation concepts, the first two are commonly used in the literature, while the third is our proposed approximation.
		\begin{itemize}
			\item[(i)]  Ignore the noise terms in \eqref{eq:NARMAX}. This yields
				\begin{equation}
					y_{\mathrm{s},k}^{(1)} = u_k - 0.1 y_{s_{k-1}}^2, \label{eq:ex3sim1}
				\end{equation}
			\item[(ii)] Truncate the NIIR after a certain number of recursive substitutions of the $y_k$ terms. We get the simulation responses
			\begin{equation}
				y_{\mathrm{s},k}^{(2)} =u_k - 0.1 u_{k-1}^2 - 0.1,\label{eq:ex3sim2}
			\end{equation}
			
			\vspace*{-0.6cm}
			\begin{equation}
				y_{\mathrm{s},k}^{(3)} =u_k - 0.1 u_{k-1}^2 - 0.1 + 0.02 u_{k-1} (u_{k-2}^2 + 1) \label{eq:ex3sim3}
			\end{equation}
			for one and two recursive substitutions, respectively, followed by truncation of the NIIR.
			\item[(iii)] After a certain number of recursive substitutions, approximate the tail of the truncated NIIR by the past sample of the simulation response. Applying this concept after one recursive replacement yields
			\begin{equation}
			y_{\mathrm{s},k}^{(4)} = \begin{multlined}[t]
				u_k - 0.1 u_{k-1}^2 - 0.1 - 0.001 (y_{s_{k-2}}^{(4)})^4 + \\
				\quad 0.02 u_{k-1} (y_{s_{k-2}}^{(4)})^2.
				\end{multlined} \label{eq:ex3sim4}
		\end{equation}
		\end{itemize}
		In Section \ref{sec:NumericalIllustration} it will be shown that the simulation model \eqref{eq:ex3sim4} approximates $E[y_k]$ better than the other models, in terms of the RMS error.
	\end{example}
	
	The purpose of approximation concept (iii) is to achieve a compromise between the approximation concept (i) and the NIIR realiztion of the simulation response. Introduce a parameter $l \in \mathbb{Z}$ to denote the number of successive recursions of the past output terms before the remaining NIIR is approximated with the past simulation response. Denote the resulting simulation approximation as $y_{\mathrm{s},k,l}$. Observe that the simulation model in \eqref{eq:ex3sim4} would be labelled as $y_{\mathrm{s},k,1}$ as per the new notation. 
	The parameter $l$ achieves a trade-off between the NIIR realization and the approximation in \eqref{eq:ex3sim1}. As $l$ approaches $\infty$, the approximation $y_{\mathrm{s},k,l}$ approaches the conditional expectation $E[y_k]$. Furthermore, for $l=0$, the approximation $y_{\mathrm{s},k,0}$ reduces to \eqref{eq:ex3sim1}.

	We now formalize the proposed approximation method for the polynomial NARMAX model class. Let $U_a^b(k) \coloneqq \{ u_{k-a},  \dots u_{k-b} \}$ be the sequence of delayed inputs with the delays ranging from $a$ to $b$ with $a,b \in \mathbb{N}_0$ and $a<b$. Similarly, define $Y_a^b(k) \coloneqq \{ y_{k-a},  \dots y_{k-b} \}$ and $\Xi_a^b(k) \coloneqq \{ \xi_{k-a},  \dots \xi_{k-b} \}$. Additionally, define a set-valued discrete-time shift operator $U_a^b(k) \diamond 1 \coloneqq U_{a+1}^{b+1}(k) = \{ u_{k-a-1},  \dots u_{k-b-1} \} $. With a slight abuse of notation, the model in \eqref{eq:NARMAX} will be re-written as
	\begin{equation}
		y_k = f\left( U_0^{n_u}, Y_1^{n_y}, \Xi_1^{n_\xi} \right) + \xi_k. \label{eq:NARMAXalt}
	\end{equation}
	The $l$ recursive substitutions can be represented in the following steps
	\begin{align}
		y_k &= f\left( U_0^{n_u}, \{ y_{k-1}, \dots y_{k-{n_y}} \}, \Xi_1^{n_\xi} \right) + \xi_k, \nonumber\\
		&= \begin{multlined}[t]
			f\big( U_0^{n_u}, \{ f(U_0^{n_u} \diamond 1, Y_1^{n_y} \diamond 1, \xi_1^{n_\xi} \diamond 1 ), \dots ,\\
			f(U_0^{n_u} \diamond n_y, Y_1^{n_y} \diamond n_y, \xi_1^{n_\xi} \diamond n_y ) \}, \Xi_1^{n_\xi} \big) + \xi_k,
			\end{multlined} \nonumber\\
		&= f_1 \left( U_0^{n_u + 1}, Y_2^{n_y + 1}, \Xi_1^{n_\xi + 1} \right) + \xi_k, \nonumber\\
		& \quad \vdots \nonumber\\
		&= f_l \left( U_0^{n_u+l}, Y_{l+1}^{n_y + l}, \Xi_1^{n_\xi + l} \right) + \xi_k. \label{eq:lSubstitutions}
	\end{align}
	It should be noted that the functions $f,f_1, \dots f_l$ are not identical. However, they are equivalent, in the sense that for  given sequences $\{u_k; k \in [0,N]\}$ and $\{\xi_k; k \in [0,N]\}$, these functions produce the same output sequence $\{y_k; k \in [\bar{n}, N] \}$, where $\bar{n} = \max(n_u+l, n_y+l, n_\xi+l)$. Moreover, since $f$ is a polynomial function, so are $f_1, \dots, f_l$. The parameterized approximation can now be defined as follows.
	\begin{definition}
		For the polynomial NARMAX model in \eqref{eq:NARMAXalt}, the $l-$\emph{approximate} simulation model is defined as
		\begin{equation}
			y_{\mathrm{s},k,l} \coloneqq E\left[f_l \left( U_0^{n_u+l}(k), Y_{s_{l+1}}^{n_y + l}(k), \Xi_1^{n_\xi + l}(k) \right) + \xi_k \right],
		\end{equation}
		where the expectation is taken with respect to the distribution of $\xi$ and input sequence $\{u_k; k \in [1,N]\}$, and ${Y_\mathrm{s}}_a^b(k) = \{ y_{\mathrm{s},k-a},  \dots y_{\mathrm{s},k-b} \}$.
	\end{definition}
	To compute the $l-$approximate simulation model, the following alternate representaiton of \eqref{eq:lSubstitutions} is used
	\begin{equation}
		y_k = \sum_{i=1}^p \left( c_i \prod_{j=0}^{n_u+l} u_{k-j}^{b_{i,j}} \prod_{r=1}^{n_y} y_{k-l-r}^{a_{i,r}} \prod_{q=1}^{n_\xi + l} \xi_{k-q}^{d_{i,q}} \right) + \xi_k, \label{eq:lSubstitutedAlt}
	\end{equation}
	where $a_{i,r}$ is the exponent of $y_{k-r}$ in the $i^\mathrm{th}$ term.
	\begin{theorem}
		The $l-$approximate simulation model for the polynomial NARMAX model in \eqref{eq:NARMAX}, under the assumption that $\xi$ is independent of $u$ and that $\xi_k \sim \mathcal{N}(0,1)$, is given by
		\vspace*{-0.25cm}
		\begin{equation}
			y_{\mathrm{s},k,l} = \sum_{i \in P_\mathrm{e}} c_i \prod_{j=0}^{n_u+l} u_{k-j}^{b_{i,j}} \prod_{r=1}^{n_y} y_{\mathrm{s},{k-l-r}}^{a_{i,m}} \prod_{q=1}^{n_\xi +l} (d_{i,q} -1)!!.
		\end{equation}
		\vspace*{-0.20cm}
		\label{th:th1}
	\end{theorem}		
	\begin{proof}
		The proof makes use of the representation in \eqref{eq:lSubstitutedAlt}, the definition of $l-$approximate simulation, and the use of Hermite polynomials of the noise process as in Section \ref{sec:NARMAXsimple}. The proof follows the same line of reasoning as in the proof of Lemma \ref{lem:1}. and is thus omitted here.
	\end{proof}
	
	\begin{remark}
		It should be noted that the proof of Lemma \ref{lem:1} critically rests on the assumption that the noise terms in \eqref{eq:NARMAX} are polynomial, and that  $\xi_k \sim \mathcal{N}(0,1)$. While Lemma \ref{lem:1} and Theorem \ref{th:th1} are derived for models of the form \eqref{eq:NARMAXsimpleAlt} and \eqref{eq:lSubstitutedAlt} respectively, the input terms $u$ and output terms $y$ were not explicitly required to be in a polynomial form. Hence, the derived results can easily be extrapolated to compute simulation models from prediction models with arbitrary non-linear functions on the input and output terms, as long as the noise terms appear in a polynomial form.
	\end{remark}
	\begin{remark}
		While we assumed that $\xi_k \sim \mathcal{N}(0,1)$, the results can be extended to a non-standard Gaussian distribution (see \eqref{eq:nonStandardGaussian}). Furthermore, the results can also be exteded to any exponential family distribution by suitably changing the probability measure and the Hilbert space of random functions in Section \ref{sec:math}. 
	\end{remark}
	

\section{Numerical illustration}
\label{sec:NumericalIllustration}

A simulation study was carried out to verify the simulation models derived in Examples \ref{ex:NARMAXsim} and \ref{ex:fullNarmax}. The system in \eqref{eq:ex2} was simulated with a periodic input $u \sim 3 + \mathcal{N}(0,1)$ with $p=1024$ periods of $N=3000$ samples, and with noise disturbance $\xi \sim \mathcal{N}(0,1)$. Data from the first 5 experiments was discarded to remove any errors due to transients. Assuming ergodicity, the ideal simulation response can be computed empirically as the ensemble average
	
		{\small
		\begin{equation}
			\bar{y}_s = \frac{1}{p-5}\sum_{i=6}^{p} y^{(i)},
		\end{equation}
		}
		
\noindent where $y^{(i)}$ is the output corresponding to the $i^\mathrm{th}$ period of the input. The theoretical simulation of the model is computed using \eqref{eq:ex2SimModel}. The empirical and the theoretical simulation response is plotted in Fig. \ref{fig:ex3Fig1}. Several realizations of the noisy output are plotted in yellow. The unscaled histogram of $\bar{y}_s - y_{\mathrm{s}}$ is depicted in Fig. \ref{fig:ex3Fig2}. The difference $\bar{y}_s - y_{\mathrm{s}}$ appears to be centered at 0. This clearly implies that the simulation model computed by neglecting the noise contributions (see \eqref{eq:ex1Sim0}) would be biased by a scalar factor of 1. The difference between the two simulation models is also quantified in terms of the RMS error in Table \ref{tab:ex2}. The first column of numbers indicate the RMS error between the empirical mean and the proposed simulation models. The second column indicates ensemble average of the RMS errors between the noisy response of the system and the simulation model approximations. It can be observed that the proposed simulation model performs significantly better in approximating the noisy output than the conventional simulation model that sets the noise contributions to 0.

\begin{figure}[t!]
	\centering
	\vspace*{0.25cm}
	\begin{subfigure}[t]{0.45\linewidth}
	\centering
		\includegraphics[scale=0.5]{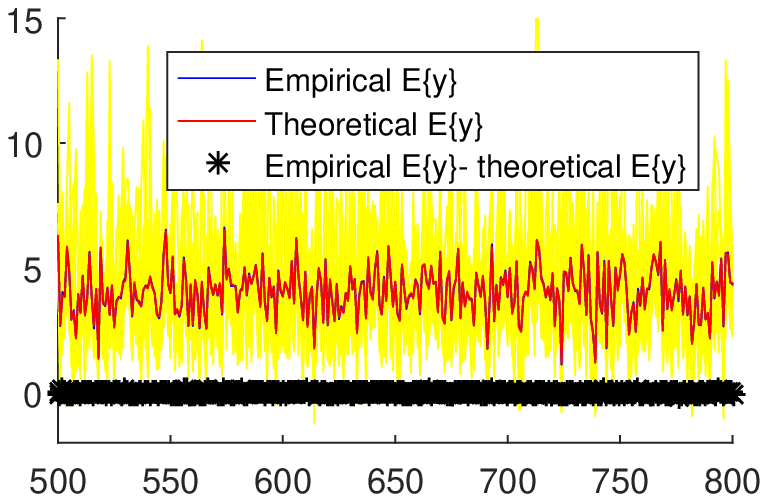}
		\caption{Comparison of $y_{\mathrm{s}}$ and $\bar{y}_s$}%
		\label{fig:ex3Fig1}
	\end{subfigure}%
	~
	\begin{subfigure}[t]{0.45\linewidth}
	\centering
		\includegraphics[scale=0.5]{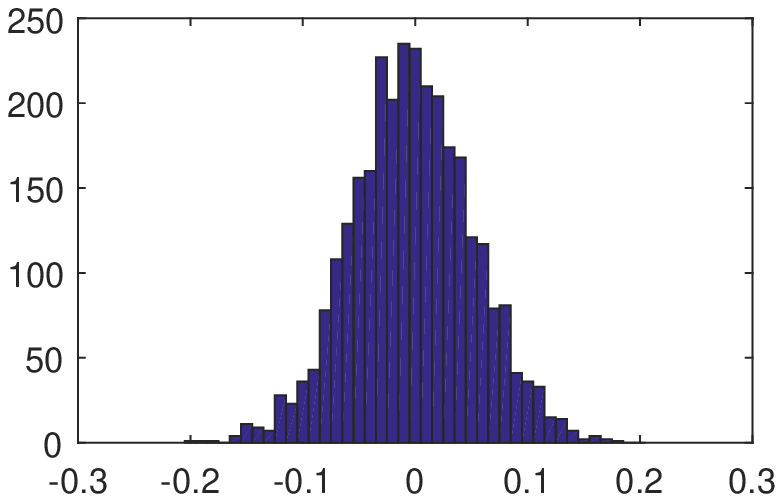}
		\caption{Distribution of $\bar{y}_s - y_{\mathrm{s}}$}
		\label{fig:ex3Fig2}
	\end{subfigure}
	\caption{Example \ref{ex:NARMAXsim} - simulation results.}
\end{figure}
\begin{table}[t!]
	\centering
	\caption{RMS errors of simulation results - Example \ref{ex:NARMAXsim}}
	\label{tab:ex2}
	\begin{tabular}{l|c|c|}
	\cline{2-3}
	\rule{0pt}{8pt}                                       & $\text{rms}(\bar{y}_s-E[y])$ & $\text{avg}_i(\text{rms}(y^{(i)} - E[y]))$ \\ \hline
	\multicolumn{1}{|l|}{$E[y]=y_{\mathrm{s}}$} & 0.0544     &  1.7305          \\ \hline
	\multicolumn{1}{|l|}{$E[y]\approx y_{\mathrm{s}}'$}      & 1.0001     & 1.9976           \\ \hline
	\end{tabular}
	\vspace*{-0.35cm}
\end{table}
A similar simulation study is carried out for the model in Example \ref{ex:fullNarmax} (see Equation \eqref{eq:ex3model}). The model was simulated with periodic input $u_k \sim \mathcal{N}(0,1)$ with $p$ periods of $N$ samples each, and noise $\xi_k \sim \mathcal{N}(0,1)$. Again, the simulation response was computed empirically using $p-5$ periods of excitation. The histograms of the differences between the four simulation models (\ref{eq:ex3sim1}-\ref{eq:ex3sim4}) and the empirical simulation $\bar{y}_s$ are depicted in Fig. \ref{fig:ex6}. Additionally, the RMS errors are presented in Table \ref{tab:ex3}. It can be verified, both visually and numerically, that $y_{\mathrm{s}}^{(4)}$ approximates the simulation of the system better than the other approximations.
\begin{figure}[t!]
		\vspace*{0.1cm}
		\centering
		\begin{subfigure}[t]{0.45\linewidth}
		\centering
			\includegraphics[scale=0.5]{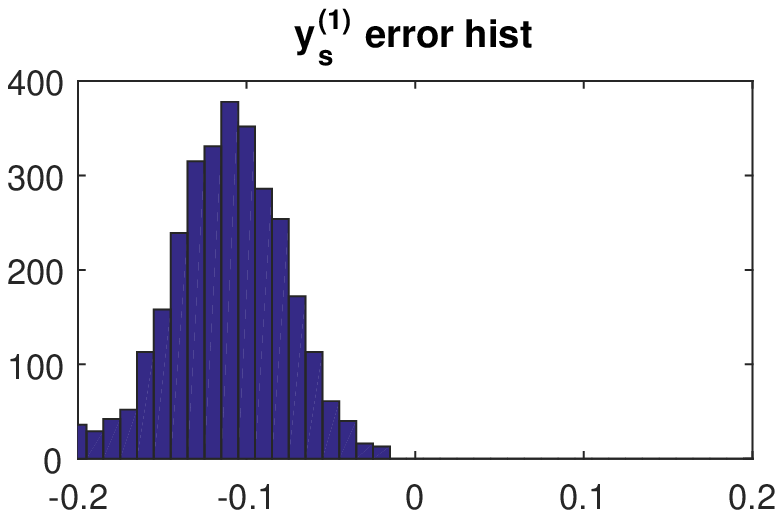}
			\caption{Distribution of $\bar{y}_s - y_{\mathrm{s}}^{(1)}$}
			\label{fig:ex6Fig1}
		\end{subfigure}%
		~
		\begin{subfigure}[t]{0.45\linewidth}
		\centering
			\includegraphics[scale=0.5]{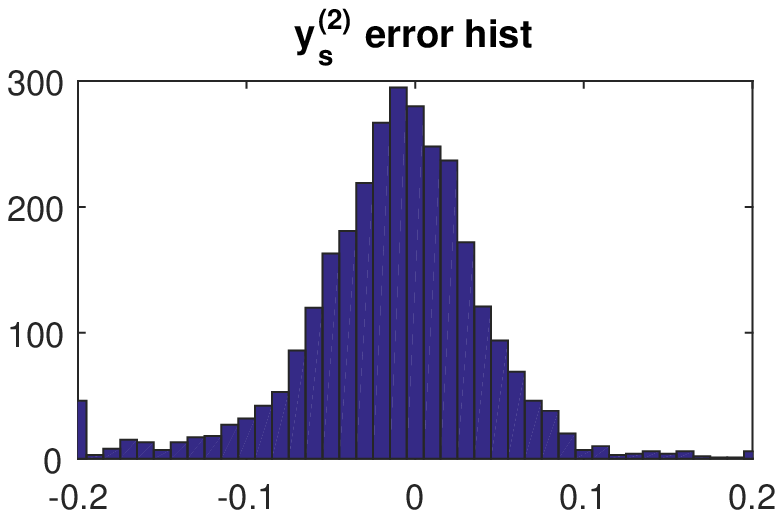}
			\caption{Distribution of $\bar{y}_s - y_{\mathrm{s}}^{(2)}$}
			\label{fig:ex6Fig2}
			\end{subfigure}%
		\\
		\begin{subfigure}[t]{0.45\linewidth}
		\centering
			\includegraphics[scale=0.5]{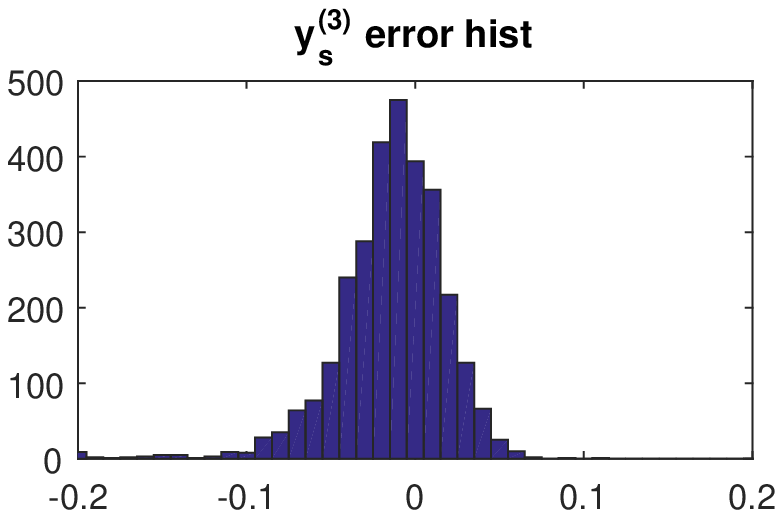}
			\caption{Distribution of $\bar{y}_s - y_{\mathrm{s}}^{(3)}$}
			\label{fig:ex6Fig3}
		\end{subfigure}%
		~
		\begin{subfigure}[t]{0.45\linewidth}
		\centering
			\includegraphics[scale=0.5]{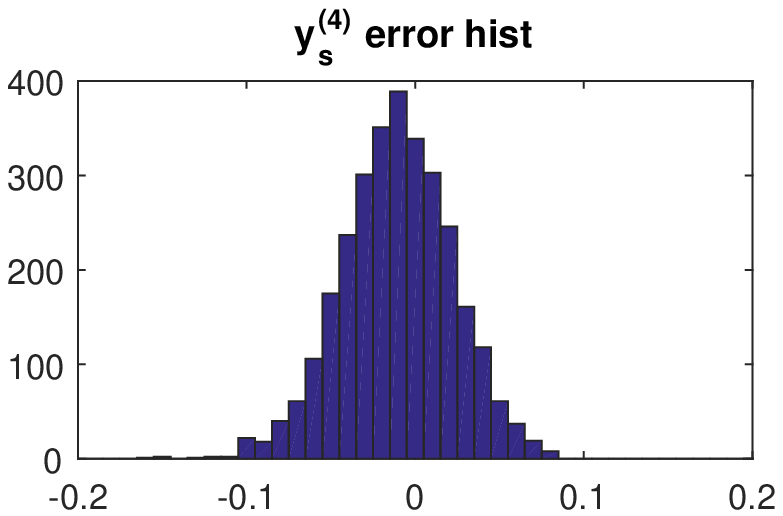}
			\caption{Distribution of $\bar{y}_s - y_{\mathrm{s}}^{(4)}$}
			\label{fig:ex6Fig4}
		\end{subfigure}
		\caption{Example \ref{ex:fullNarmax} - simulation results.}
		\label{fig:ex6}
		\vspace*{-0.7cm}
\end{figure}
%
\begin{table}[t!]
	\centering
	\vspace{0.2cm}
	\caption{RMS errors of simulation results - Example \ref{ex:fullNarmax}}
	\label{tab:ex3}
	\begin{tabular}{l|c|c|}
	\cline{2-3}
	\rule{0pt}{8pt}                                       & $\text{rms}(\bar{y}_s-E[y])$ & $\text{avg}_i(\text{rms}(y^{(i)} - E[y]))$ \\ \hline
	\multicolumn{1}{|l|}{$E[y] \approx y_{\mathrm{s}}^{(1)}$} & 0.1165             & 1.0407                \\ \hline
	\multicolumn{1}{|l|}{$E[y] \approx y_{\mathrm{s}}^{(2)}$} & 0.0618             & 1.036                 \\ \hline
	\multicolumn{1}{|l|}{$E[y] \approx y_{\mathrm{s}}^{(3)}$} & 0.0360             & 1.0348                \\ \hline
	\multicolumn{1}{|l|}{$E[y] \approx y_{\mathrm{s}}^{(4)}$} & 0.0348             & 1.0348                \\ \hline
	\end{tabular}
	\vspace*{-0.7cm}
\end{table}

\section{Conclusions}

When computing a simulation model from a prediction model that belongs to the NARMAX class, care must be taken to avoid bias. To that extent, we proposed an alternative representation that allows to compute a simulation model from a polynomial NARMAX model conveniently. In order to curtail infinite order simulation models, we also proposed a parameterized method to approximate the simulation model. The proposed approximation method yields the true, infinite-order simulation model as the approximation parameter goes to $\infty$. Simulation examples were used to demonstrate that the approximation method yields a simulation model with a smaller bias compared to the other common approaches used for computing simulation models.


\balance

\bibliographystyle{IEEEtran}

\bibliography{cdc2018_bib}

\end{document}